\documentclass[a4paper]{article} 

\usepackage[left=2cm,right=12cm]{geometry}

\usepackage{xr}
\usepackage{epsfig}
\usepackage{color}
\usepackage{graphicx}
\usepackage{soul}
\usepackage{amsmath,amsfonts,amssymb}
\usepackage{setspace} 
\usepackage{colortbl}

\definecolor{lightyellow}{rgb}{1,0.98,0.8}
\definecolor{orange}{rgb}{1,0.75,0.3}
\definecolor{brown}{rgb}{0.95,0.6,0.3}
\definecolor{pink}{rgb}{1,0.7,0.7} 
\definecolor{purple}{rgb}{0.8,0.7,1}
\definecolor{lightblue}{rgb}{0.65,0.7,1}

\definecolor{c1}{rgb}{0,0,0.8}
\definecolor{c2}{rgb}{0.25,0.,0.6}
\definecolor{c3}{rgb}{0.6,0.,0.3} 
\definecolor{c4}{rgb}{1,0,0}
\definecolor{c5}{rgb}{.5,0.2,0}

\definecolor{lightblack}{rgb}{.5,0.1,0.2}

\usepackage[colorinlistoftodos]{todonotes}

\newcommand{\metabolitepolytope}{{\mathcal P}}

\newtheorem{lemma}{Lemma}

 \usepackage{amsmath,amsfonts,amssymb}

\newcommand{\inv}{^{-1}}

\newcommand{\diag}{\mbox{\rm Dg}}

\newcommand{\Rset}{{\mathbb R}}

\newcommand{\e}{\mbox{\rm e}}

\newcommand{\trans}{^{\top}}

\newtheorem{theorem}{Theorem}[section]

\newtheorem{proof}[theorem]{Proof}

\usepackage{epsfig}
\usepackage{graphicx}

\newcommand{\myparagraph}[1]{\vspace{-3mm}\paragraph{#1}}

\renewcommand{\diag}{\mbox{Dg}}
\newcommand{\sign}{\mbox{sign}}

\newcommand{\Ntot}{{{\bf N}^{\rm tot}}}
\newcommand{\Nint}{{{\bf N}}}

\newcommand{\ratelaw}{r}

\newcommand{\rrs}{h^{\rm v}}
\newcommand{\rrl}{h^{{\rm v}}_{l}}

\newcommand{\cv}{{\bf c}}
\newcommand{\vv}{{\bf v}}

\newcommand{\uv}{{\bf u}}

\newcommand{\kv}{{\bf k}}

\newcommand{\av}{{\bf a}}

\newcommand{\muv}{{\boldsymbol \mu}}

\newcommand{\ev}{{\bf e}}
\newcommand{\xv}{{\bf x}}

\newcommand{\nv}{{\bf n}}

\newcommand{\Esc}  {{\bf E}}

\newcommand{\wolfmore}[1]{}
\newcommand{\wolftodo}[1]{}
\newcommand{\elad}[1]{}
\newcommand{\eladtodo}[1]{}
\newcommand{\avi}[1]{}
\newcommand{\avitodo}[1]{}
\renewcommand{\myparagraph}[1]{}

\newcommand{\yy}{y}

\renewcommand{\Esc}{E'}
\newcommand{\enzyme}{\varepsilon}
\newcommand{\el}{\varepsilon_{l}}
\renewcommand{\ev}{{\boldsymbol \varepsilon}}

\newcommand{\zv}{{\bf z}}

\newcommand{\hEl}{h_{l}}
\newcommand{\hE}{h}
\newcommand{\concS}{s}
\newcommand{\concP}{p}

\newcommand{\kcatplus}{k^{+}_{\rm cat}}
\newcommand{\kcatminus}{k^{-}_{\rm cat}}
\newcommand{\kcat}{k^{\rm cat}}
\newcommand{\kcatl}{k^{\rm cat}_l}

\newcommand{\kM}{K^{{\rm M}}}
\newcommand{\kMS}{K_{{\rm S}}}
\newcommand{\kMP}{K_{{\rm P}}}
\newcommand{\kMi}{K^{{\rm M}}_{i}}
\newcommand{\kMj}{K^{{\rm M}}_{j}}
\newcommand{\kMli}{K^{\rm M}_{li}}
\newcommand{\keq}{K_{\rm eq}}
\newcommand{\fluxcost}{y^{\rm opt}}
\newcommand{\regterm}{y^{\rm reg}}

\renewcommand{\rrs}{\Upsilon}
\renewcommand{\rrl}{\rrs_{l}}

\newcommand{\rrlmin}{\rrs_{l}^{\rm (1)}}

\renewcommand{\ratelaw}{r}

\begin{document}

\title{The enzyme cost of given metabolic flux distributions, as a function of logarithmic metabolite levels, is convex}

\date{} \author{Wolfram Liebermeister$^{1}$ and Elad Noor$^{2}$\\
$^{1}$ Institut f\"ur Biochemie, Charit\'e - Universit\"atsmedizin Berlin\\
$^{2}$ Institute of Molecular Systems Biology, ETH Z\"urich}

\maketitle

\begin{abstract} 
Enzyme costs play a major role in the choice of metabolic routes, both
in evolution and bioengineering. Given desired fluxes, necessary
enzyme levels can be estimated based on known rate laws and on a
principle of minimal enzyme cost.  With logarithmic metabolite levels
as free variables, enzyme cost functions and constraints in optimality
and sampling problems can be handled easily.  The set of feasible
metabolite profiles forms a polytope in log-concentration space, whose
points represent all possible steady states of a kinetic model.  We
show that  enzyme cost is a convex function on this polytope. This
makes enzyme cost minimization -- finding  optimal enzyme
profiles and corresponding metabolite profiles that realize a
desired flux at a minimal cost -- a convex optimization problem.
\end{abstract}

\section{Introduction}

\myparagraph{Enzyme cost minimization} The metabolic fluxes in cells
are driven by enzymes, which come at a cost.  Translating a given flux
profile into the necessary enzyme profile, and computing the
corresponding enzyme cost, is not a trivial task. In kinetic models, a
reaction rates $v = E \cdot \ratelaw(\cv)$ is a product of enzyme
level $E$ and an enzyme-specific rate given by the rate law
$\ratelaw(\cv)$. If metabolite levels were known, the enzyme demand
could be directly computed; the specific enzyme demand $E/v =
1/\ratelaw(\cv)$ is simply obtained by inverting the rate law.
However, since metabolite levels are not fixed, the fluxes in a
network can be realized by many possible enzyme profiles, each with a
corresponding metabolite profile.  To select a plausible  solution, we employ
an optimality principle: we define an enzyme cost function (for
instance, total enzyme mass) and choose among all possible enzyme
profiles the one with the lowest cost. As a constraint, the
corresponding metabolite profile must respect physiological ranges and
energetic constraints implied by flux directions.
\myparagraph{Separable rate laws and separable enzyme cost functions}
The enzyme demand in a reaction, at a given desired flux, depends on
thermodynamic and kinetic factors.  To see what each factor
contributes, we split the formula for enzyme demands into a product of
terms, each with a simple interpretation.  The reaction rate depends
on enzyme level, forward catalytic constant $\kcatplus$ (i.e., the
maximal possible forward rate per mM of enzyme), driving force (which
determines the relative backward flux), and kinetic effects (such as
substrate saturation or allosteric regulation) that modify the forward
flux (see Figure \ref{fig:efficiencies}).

\begin{figure}[t!]
  \begin{center}
    \includegraphics[width=15cm]{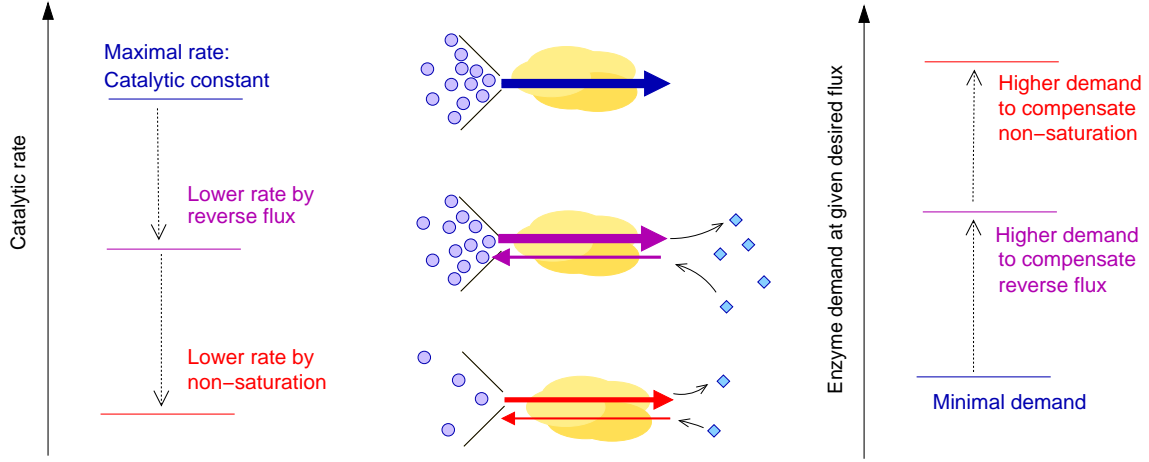}
  \end{center}
  \caption{The catalytic rate of enzymes is decreased by different
    physical factors.  Under ideal conditions, an enzyme molecule
    catalyses its reaction at a maximal rate, given by the enzyme's
    forward catalytic constant (top). The rate is reduced by
    microscopic reverse fluxes (center) and insufficient availability
    of substrate (incomplete saturation, leading to waiting times
    between conversion events). As the catalytic rate of the enzyme
    decreases (left), realizing a desired metabolic flux requires
    increasingly more enzyme (right).}
  \label{fig:efficiencies}
\end{figure}

\section{Enzymatic rate laws}

\myparagraph{Rate laws for general enzymatic reactions} Reactions of
the form A $\leftrightharpoons$ B can be described by the reversible
Michaelis-Menten kinetics.  A generalized form for
reactions with multiple substrates (concentrations $s_{i}$) and
products (concentrations $p_{j}$) reads
\begin{eqnarray}
  \label{eq:GeneralRateLawRate}
v &=&  E\, \frac{
k^+_{\rm cat} \prod_{i} (\frac{s_{i}}{\kMi})^{m^{\rm S}_{i}}
- k^-_{\rm cat} \prod_{i} (\frac{p_{i}}{\kMi})^{m^{\rm P}_{i}}
}{D(s_1, s_2, ..,p_1, p_2, ..)}.
\end{eqnarray}
The molecularities $m^{\rm S}_{i}$ and $m^{\rm P}_{i}$ represent the
(positive) stoichiometric coefficients, but they may be scaled by a
reaction-specific factor which effectively acts like a Hill
coefficient. Using a stoichiometric coefficient $n_{il}$ and a
  molecularity $m^{\rm S}_{li} = 2 |n_{il}|$ is equivalent to using a
  Hill coefficient of 2 in the rate law.  For reasons of thermodynamic consistency,
reaction rates must vanish in chemical equilibrium states; to ensure
this, equilibrium constants and rate constants must satisfy the
Haldane relationship \cite{hald:30}
\begin{eqnarray}
  \label{eq:HaldaneRelationship}
 \keq =  \frac{\prod_{i} (s^{\rm eq}_{i})^{m^{\rm S}_{i}}}
{\prod_{i} (p^{\rm eq}_{i})^{m^{\rm P}_{i}}}
=  \frac{  k^+_{\rm cat} \prod_{i} (\kMi)^{m^{\rm P}_{i}}}
{k^-_{\rm cat}  \prod_{i} (\kMi)^{m^{\rm S}_{i}}},
\end{eqnarray}
where $s_i$ and $p_j$ denote to substrate and product levels,
respectively.  Since the equilibrium constants depend on the Gibbs
energies of formation as $\keq = \e^{-\Delta_{\rm r}
  {G^{\circ}}'/RT}$, they must satisfy Wegscheider conditions
\cite{wegs:02}: the vector of equilibrium constants satisfies $\ln
\keq = {\Ntot}\trans \,{\muv^\circ}'$, with the stoichiometric matrix
$\Ntot$ for all metabolites and the vector ${\muv^\circ}'$ of
transformed Gibbs free energies of formation. Accordingly, the
equilibrium constants must satisfy a Wegscheider condition $\ln \keq
\cdot \kv = 0$ for any thermodynamic cycle $\kv$, i.e., any nullspace
vector of ${\Ntot}\trans$.  The denominator $D$ in
Eq.~(\ref{eq:GeneralRateLawRate}) depends on the enzyme mechanism.  In
general, it is a polynomial
\begin{eqnarray}
\label{eq:DECF4}
D(\cv) = 1 + \sum_{k} M_{lk} \prod_i c_{i}^{m_{lik}}
\end{eqnarray}
of the metabolite concentrations with positive coefficients $M_{lk}$
and exponents $m_{lik}$. For examples of such denominators, see
appendix \ref{sec:SIratelaws}.  In the underlying enzyme mechanism,
each sum term (index $k$) represents a binding state of the enzyme.
The exponents $m_{lik}$ indicate the numbers of reactant molecules
bound in one binding state and the prefactors encode the binding
energies. The sum term 1 represents the unbound enzyme.  The
highest-order substrate term reads $\prod_{i} (s_{i}/\kMi)^{m^{\rm
    S}_{i}}$ and the highest-order product term reads $\prod_{i}
(p_{i}/\kMi)^{m^{\rm P}_{i}}$.  In addition, the denominator may
contain additive or multiplicative terms for allosteric activation and
inhibition.  The exponents $m_{lik}$ are usually positive integer
numbers.  With allosteric regulation, however, there can also be
denominator terms of the form $\kMS/s$.

\section{Separable rate laws and enzyme cost}

Following \cite{nflb:13}, we consider general reversible rate laws and
factorize them into
\begin{eqnarray}
\label{eq:factorised}
  v &=& \enzyme \cdot \kcat \cdot \eta^{\rm th} \cdot \eta^{\rm kin} \cdot \eta^{\rm reg}, 
\end{eqnarray}
where $\kcat = \kcatplus$ is the forward catalytic constant.  For an
example of such a factorization, see appendix
\ref{sec:SIfactorization}.  The energetic efficiency
\begin{eqnarray}
\label{eq:factorisedThermo}
\eta^{\rm th} &=& 1-\frac{\Gamma}{\keq} = 1 - \e^{-\Theta}
\end{eqnarray} 
depends on the mass-action ratio $\Gamma$ (e.g., $\Gamma =
\concP/\concS$ for unimolecular reactions) and on the equilibrium
constant $\keq$, or briefly on the driving force $\Theta =
-\Delta_{\rm r} G/RT$. Note that our driving forces are defined via
molecularities, not via the stoichiometric coefficients; to allow for
a consistent equilibrium state, all reactants within a reaction must
show the same Hill coefficient \cite{liuk:10}. The relationship
$\Gamma/\keq=\e^{-\Theta}$ links concentrations to driving forces and
holds for ideal chemical mixtures with constant activity
coefficients. The kinetic efficiency $\eta^{\rm kin}$ depends on the
rate law and can be derived from the rate law's denominator.  
For a general reversible rate law, the kinetic efficiency would read
 \begin{eqnarray}
   \label{eq:GeneralRateLawMulti}
   \eta^{\rm kin} &=& \frac{\prod_{i} (s_{i}/\kMi)^{m^{\rm S}_{i}}}{D(s_1, s_2, .., p_1, p_2, ..)} 
 \end{eqnarray}
 where the substrate-dependent numerator $\prod_{i}
 (s_{i}/\kMi)^{m^{\rm S}_{i}}$ stems from the positive numerator term
 in the rate law and the denominator $D$ is given by the rate law
 denominator.  A factorized formula (\ref{eq:factorised}), called
 separable rate law \cite{nflb:13}, exists for reactions of arbitrary
 stoichiometry (for examples, see SI \ref{sec:SIratelaws}).  The
 factorization is always possible even for general rate laws, because
 rate law numerators must have the form $\kcatplus \prod_i\,
 (s_i/{\kM}_{li})^{m^{\rm S}_{li}} - \kcatminus\, \prod_i
 (p_i/{\kM}_{li})^{m^{\rm P}_{li}}$ for reasons of thermodynamic
 consistency.

\begin{figure}[t!]
\colorbox{lightyellow}{\begin{minipage}{15cm}
\hspace{1cm}\parbox{14cm}{\vspace{5mm} (a)  Reversible
    Michaelis-Menten kinetics (factorized, with driving force $\theta =
    -\Delta_{\rm r} G/RT$)

\[v = {\color{c1}\enzyme}
\cdot {\color{c2}\kcatplus}
\cdot \underbrace{{\color{c3}[1-\mbox{e}^{-\theta}]}}_{\eta^{\rm th}}
\cdot \underbrace{{\color{c4}\frac{\concS/\kMS}{1+\concS/\kMS+\concP/\kMP}}}_{\eta^{\rm kin}}
\cdot \underbrace{{\color{c4}\frac{1}{1+x/K_{\rm I}}}}_{\eta^{\rm reg}}\]

Rate = {\color{c1} Enzyme level} $\cdot$ {\color{c2}Forward catalytic constant} $\cdot$ {\color{c3}Energetic efficiency} $\cdot$ {\color{c4}Kinetic efficiency}\ \\

(b) Enzyme cost function (factorized form)

\[
\yy =  
          {\color{c5}h} \cdot {\color{c1}\enzyme} = 
          {\color{c5}h} \cdot v
          \cdot {\color{c2}\frac{1}{\kcatplus}}
          \cdot \underbrace{{\color{c3}\frac{1}{[1-\mbox{e}^{-\theta}]}}}_{1/\eta^{\rm th}}
          \cdot \underbrace{{\color{c4}\frac{1+\concS/\kMS+\concP/\kMP}{\concS/\kMS}}}_{1/\eta^{\rm kin}}
          \cdot \underbrace{{\color{c4} [1+x/K_{\rm I}]}}_{1/\eta^{\rm reg}}
          \]
}
\end{minipage}
}
  \caption{Separable rate law and enzyme cost function.  (a)
    Reversible rate laws can be factorized \cite{nflb:13}.  The
    example shows a reaction S $\leftrightharpoons$ P with reversible
    Michaelis-Menten kinetics Eq.~(\ref{eq:mmratelaw}) and a 
    factor for non-competitive allosteric inhibition (inhibitor
    concentration $x$).  (b) The enzyme cost $y$ (enzyme level
    $\enzyme$, multiplied by the specific enzyme cost $\hE$) contains
    the terms from the rate law in inverse form. By omitting some
     terms (or replacing them by constant numbers), one obtains
    simplified enzyme cost functions.  }
  \label{fig:pscscores}
\end{figure}

The terms in the rate law (\ref{eq:factorised}) depend on metabolite
levels in different ways.  The first terms, $\enzyme \cdot \kcatplus$,
represent the maximal velocity (the rate at full substrate-saturation,
no backward flux, full allosteric activation), while the following
efficiency terms describe how this velocity is reduced in reality: the
factor $\eta^{\rm th}$ describes a reduction due to backward fluxes,
and the factors $\eta^{\rm kin}$ and $\eta^{\rm th}$ describe a
further reduction due to incomplete substrate saturation and
allosteric regulation.  While $\kcatplus$ is an enzyme-specific
constant (yet, dependent on conditions such as pH, ionic strength, or
molecular crowding in cells; unit 1/s), the efficiency terms are
concentration-dependent, unitless, and can vary between 0 and 1. The
thermodynamic efficiency $\eta^{\rm th}$ depends on the driving force
(and thus, indirectly, on metabolite levels) and the equilibrium
constant is required for its calculation. The kinetic efficiency
$\eta^{\rm kin}$ depends directly on metabolite levels and contains
the $\kM$ values as parameters. Allosteric regulation can be captured
by $\eta^{\rm kin}$ (as additive or multiplicative terms in the
denominator), but non-competitive allosteric regulation can also be
described by a separate term $\eta^{\rm reg}$.  \myparagraph{Enzyme
  cost functions} If rate law, flux, and metabolite levels are known,
a reaction's enzyme demand follows from Eq.~(\ref{eq:factorised}) as
\begin{eqnarray}
\label{eq:factorisedU}
\enzyme_l(v, \cv) &=& v_l \cdot \frac{1}{\kcatl} \cdot
\frac{1}{\eta^{\rm th}_l(\Theta(\cv))}
 \cdot \frac{1}{\eta_l^{\rm kin}(\cv)} \cdot \frac{1}{\eta_l^{\rm reg}(\cv)}.
\end{eqnarray} 
By weighting the enzyme demand with an enzyme-specific
cost $\hEl$, we obtain the cost function
\begin{eqnarray}
\label{eq:ECF}
y_{l}(v_{l},\cv) &=& \underbrace{h_{l}\cdot \frac{1}{\kcatl}}_{\rrlmin} \cdot
\underbrace{\frac{1}{\eta^{\rm th}_l(\cv)}}_{[1-\e^{-\Theta_l(\cv)}]\inv}
\cdot \frac{1}{\eta_l^{\rm kin}(\cv)}
\cdot \frac{1}{\eta^{\rm reg}(\cv)}  \cdot v_{l} =\rrl(\cv)\,v_l
\end{eqnarray}
Dividing Eq.~(\ref{eq:ECF}) by $v_{l}$, we obtain the \emph{specific
  flux cost} $\rrl=y_{l}/v_{l} = \hEl/\ratelaw_l$.  Eq.~(\ref{eq:ECF})
shows which factors shape enzyme cost, and how. The first two terms
yield the minimal cost $\rrlmin= h_{l}/\kcatl$ (the cost per flux
under ideal conditions); the following terms further increase this
value.  To keep enzyme cost low, the $\kcat$ values, driving forces
(i.e.~imbalance in substrate and product levels), and substrate
saturation (i.e.~high substrate concentrations) should be high.  For a
pathway with desired fluxes $v_l$ and known log-concentrations $x_i =
\ln c_i$, the total cost reads
\begin{eqnarray}
\label{eq:TotalEnzymeDemand}
\yy_{\rm pw}(\xv) &=& \sum_{l} y_{l}(\vv, \cv) = \sum_{l} \hEl\,\el(v_{l},\cv).
\end{eqnarray}
Setting $\hEl=1$, we obtain the total enzyme demand (as a concentration
in mM), and setting $\hEl=m_l$ (protein mass in Daltons), we obtain
the mass concentration (in gram protein per cell dry weight) as
a special case.

\begin{figure}[t!]
  \begin{center}
    \begin{tabular}{l}
      (a)\\[2mm]
      \includegraphics[width=1.8cm]{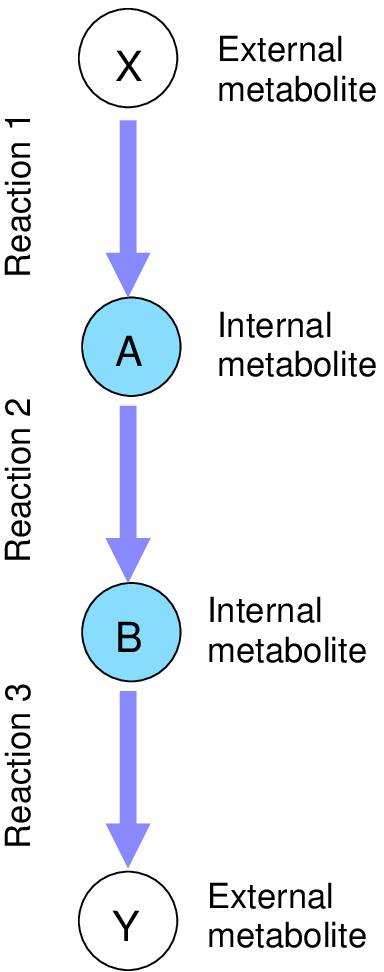}
    \end{tabular}
    \begin{tabular}{llll}
      (b)& (c)& (d)&(e) \\
      \includegraphics[height=2.8cm]{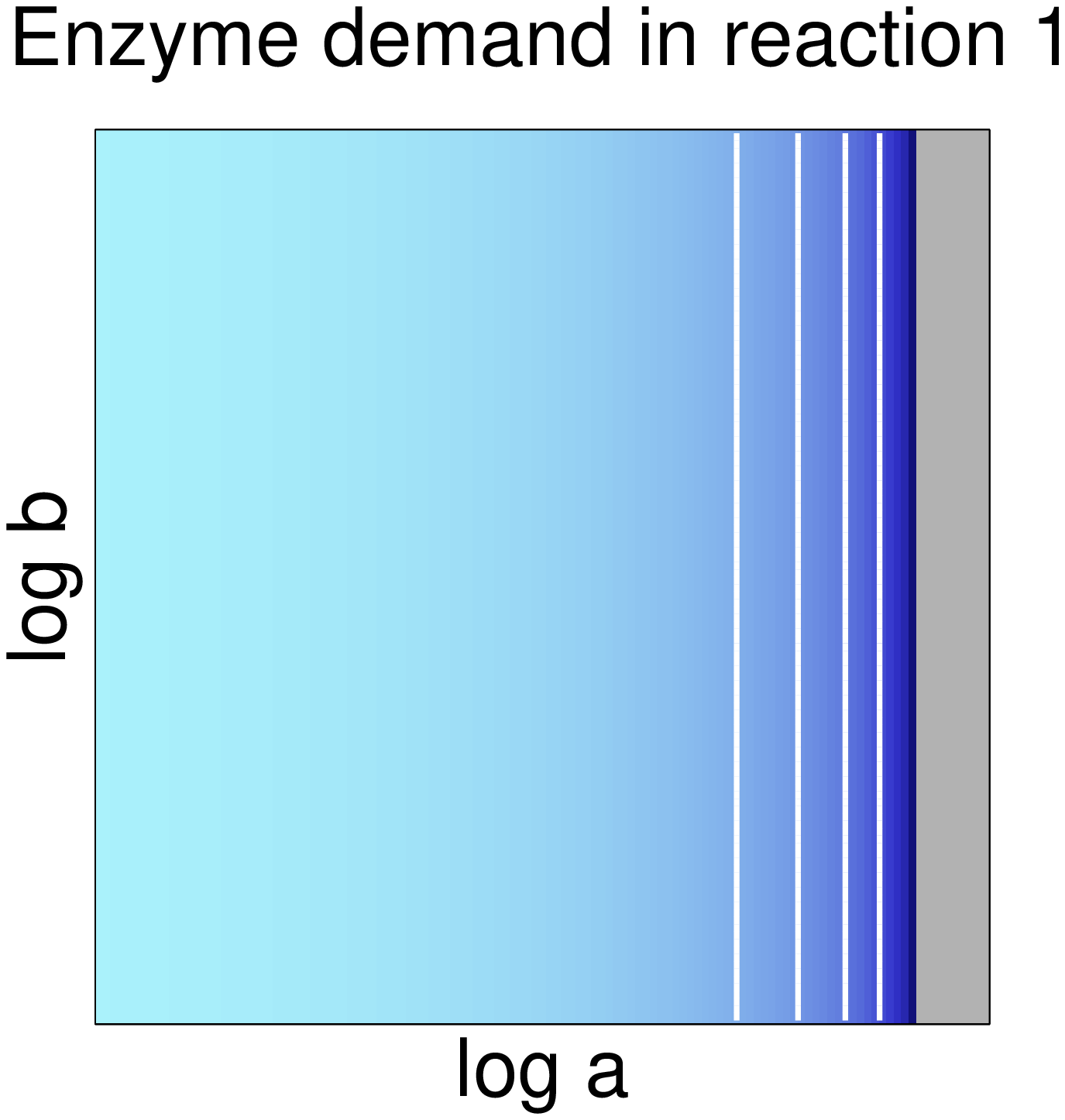} & 
      \includegraphics[height=2.8cm]{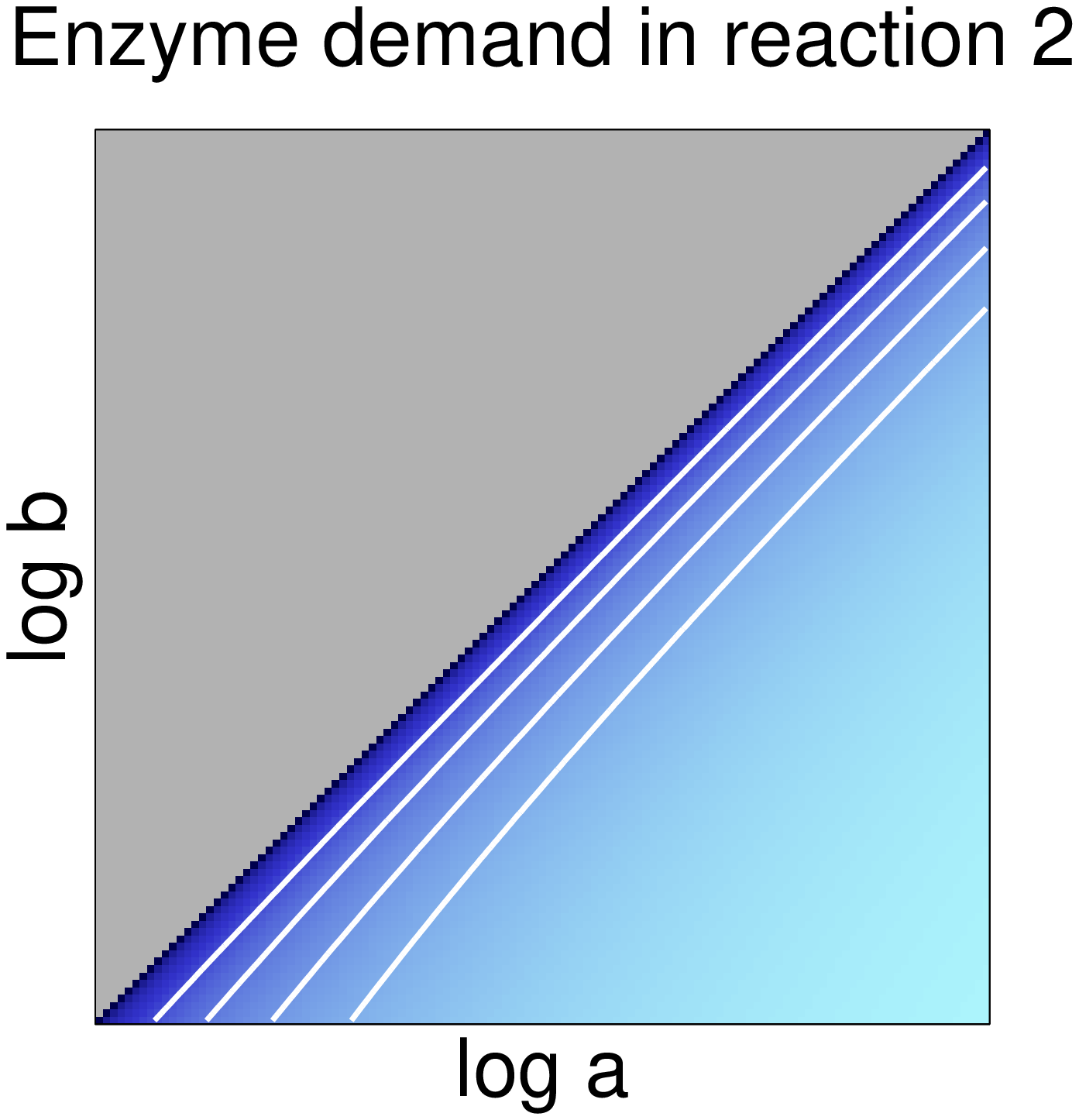}&
      \includegraphics[height=2.8cm]{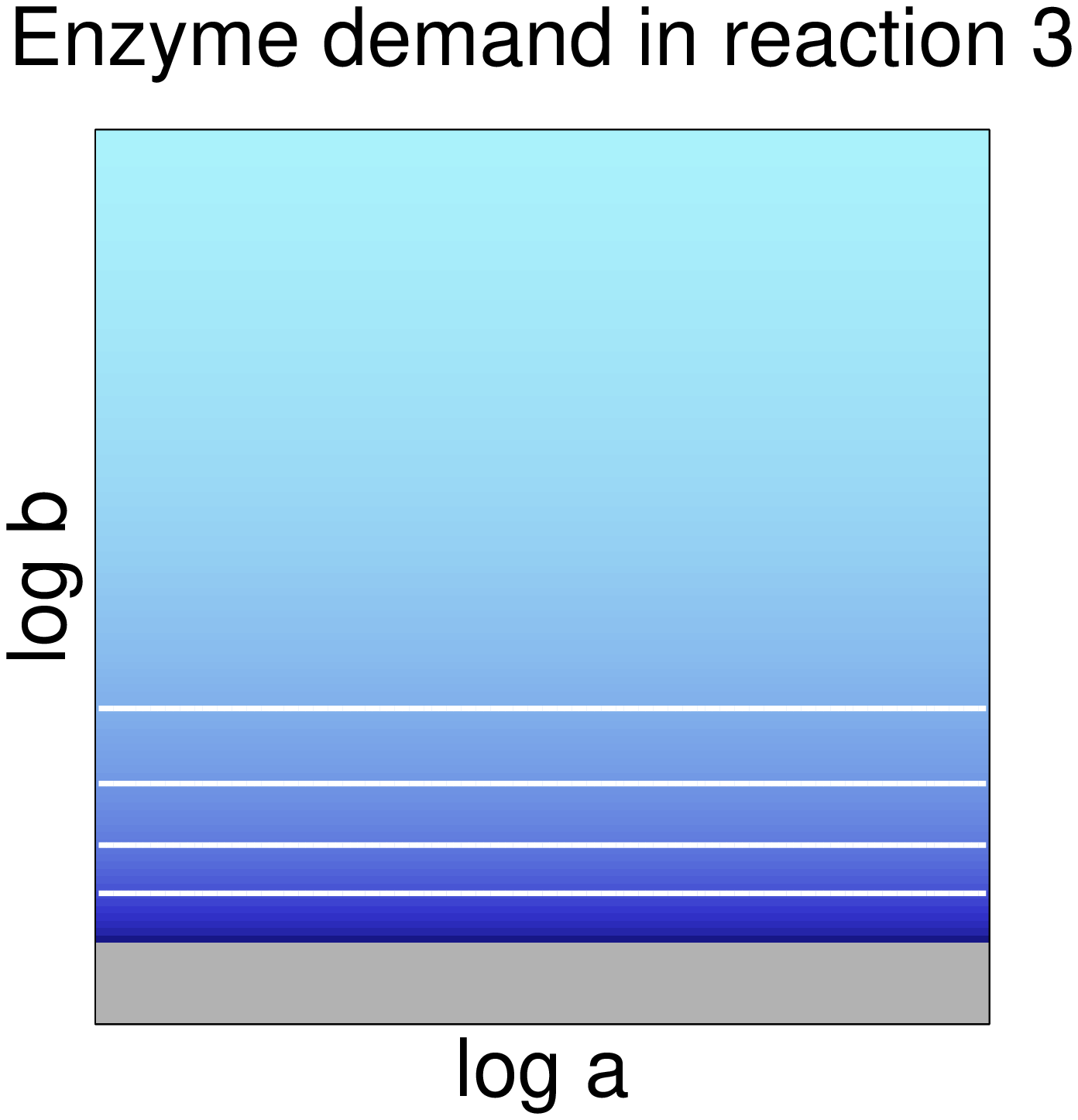}& 
      \includegraphics[height=2.8cm]{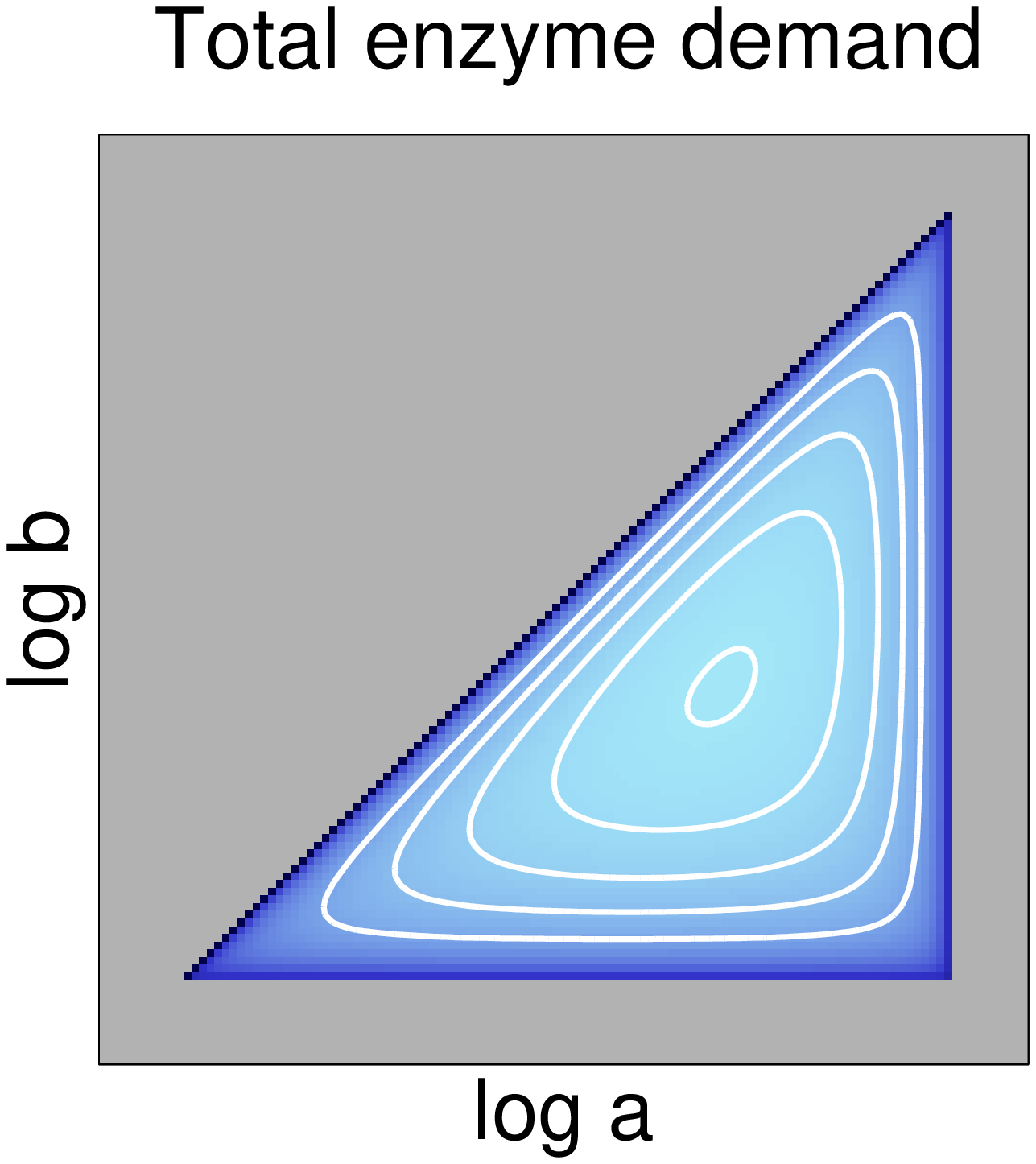}\\
      \includegraphics[width=2.8cm]{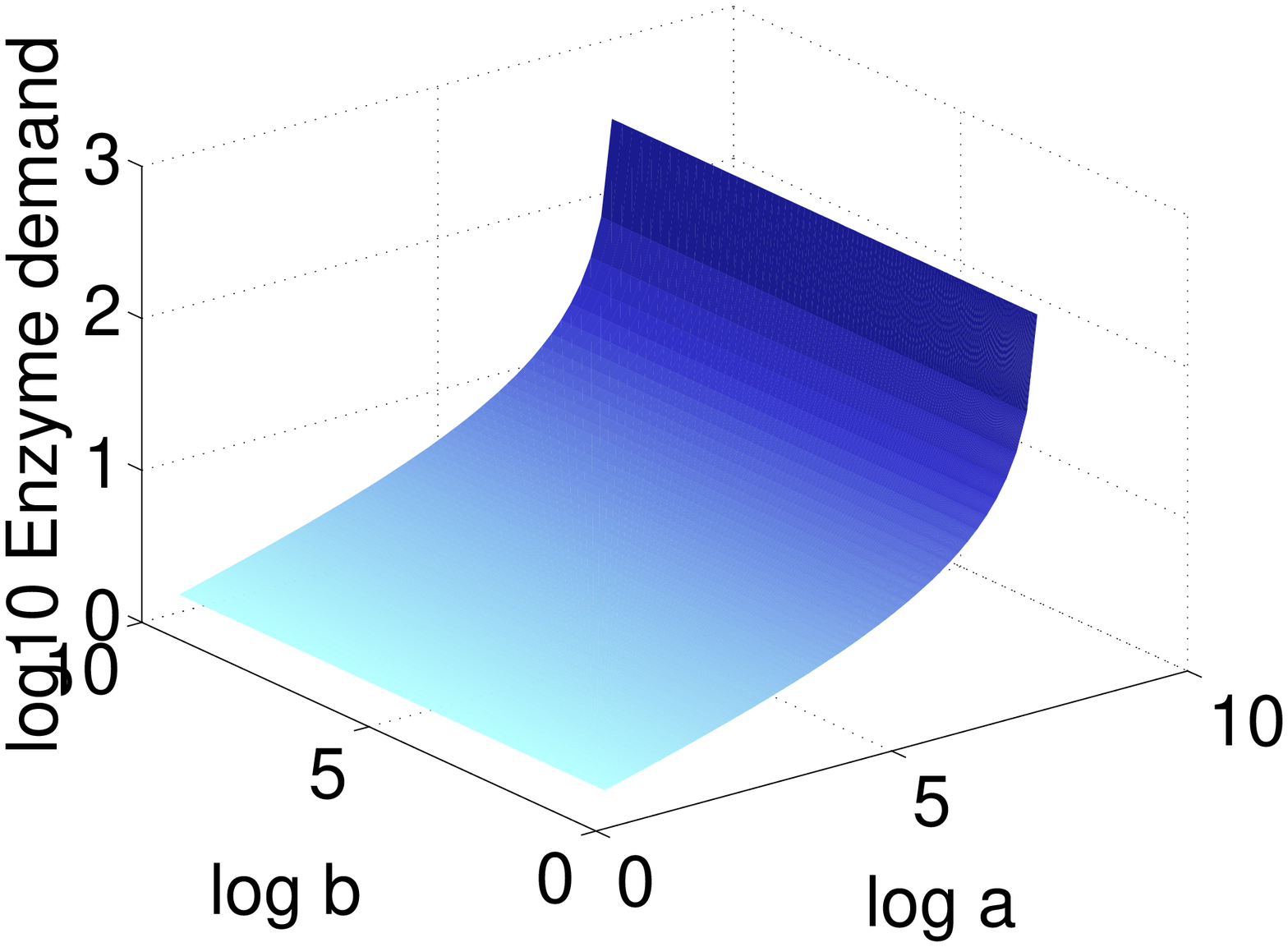} & 
      \includegraphics[width=2.8cm]{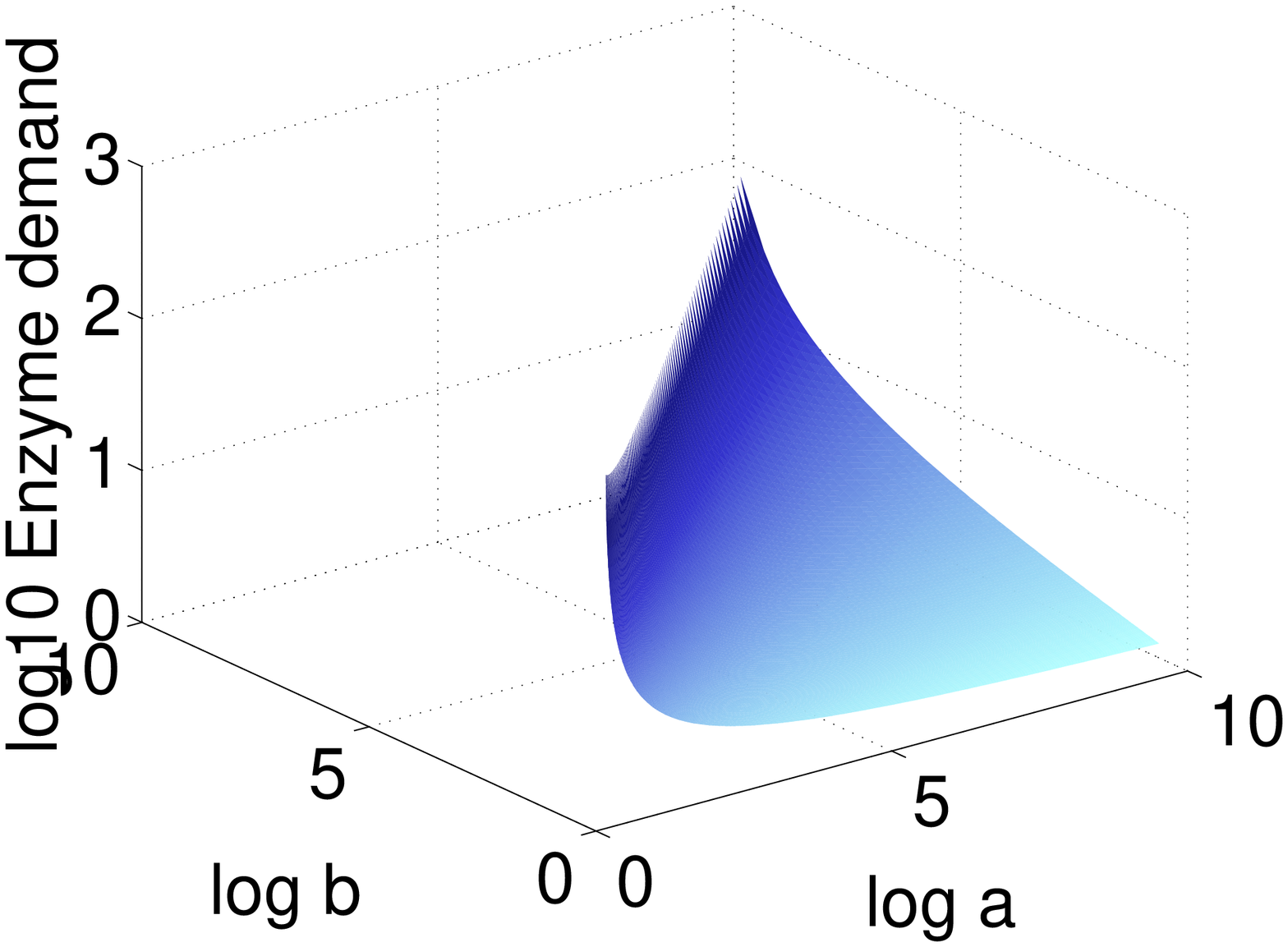}&
      \includegraphics[width=2.8cm]{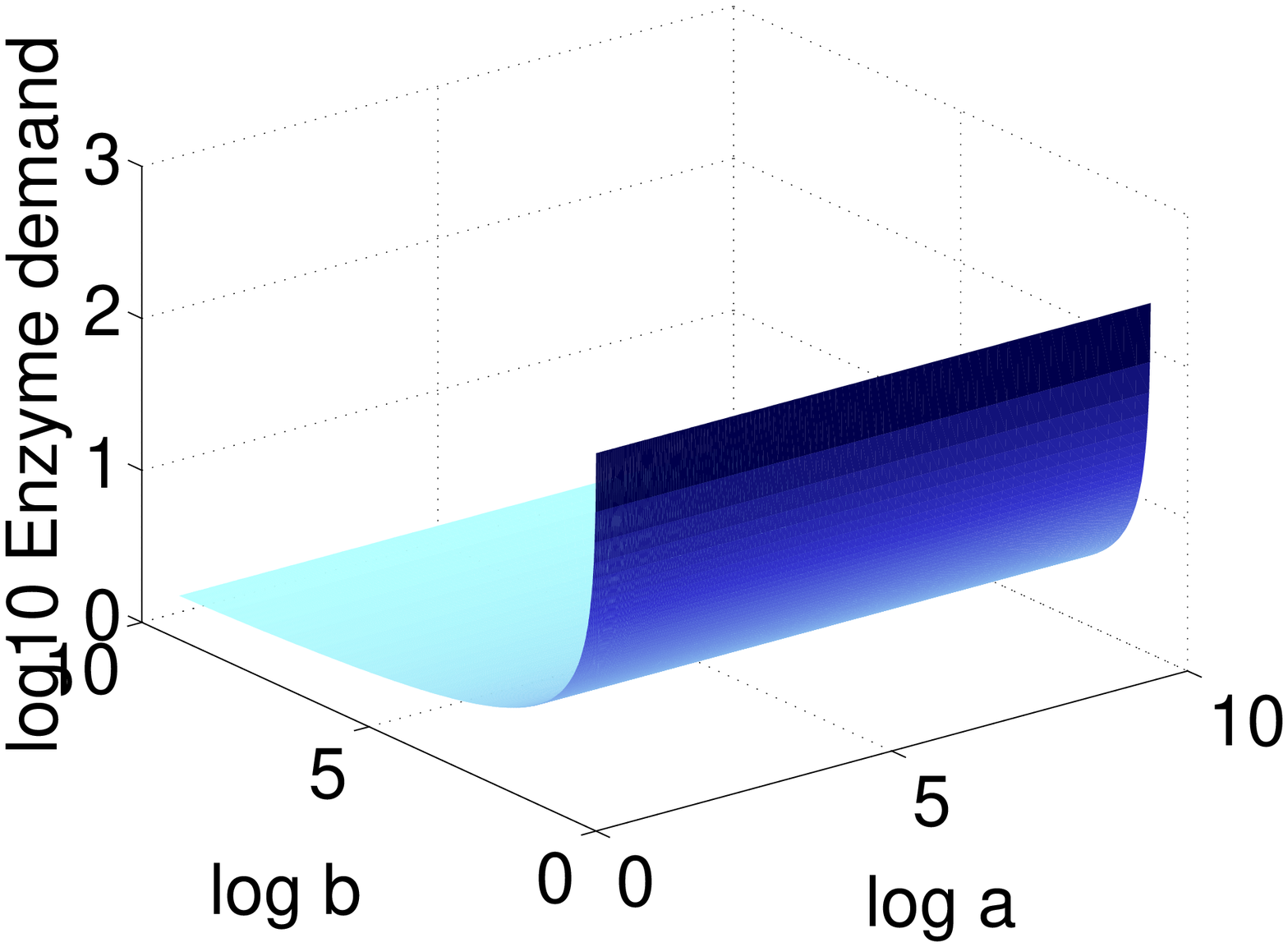}& 
      \includegraphics[width=2.8cm]{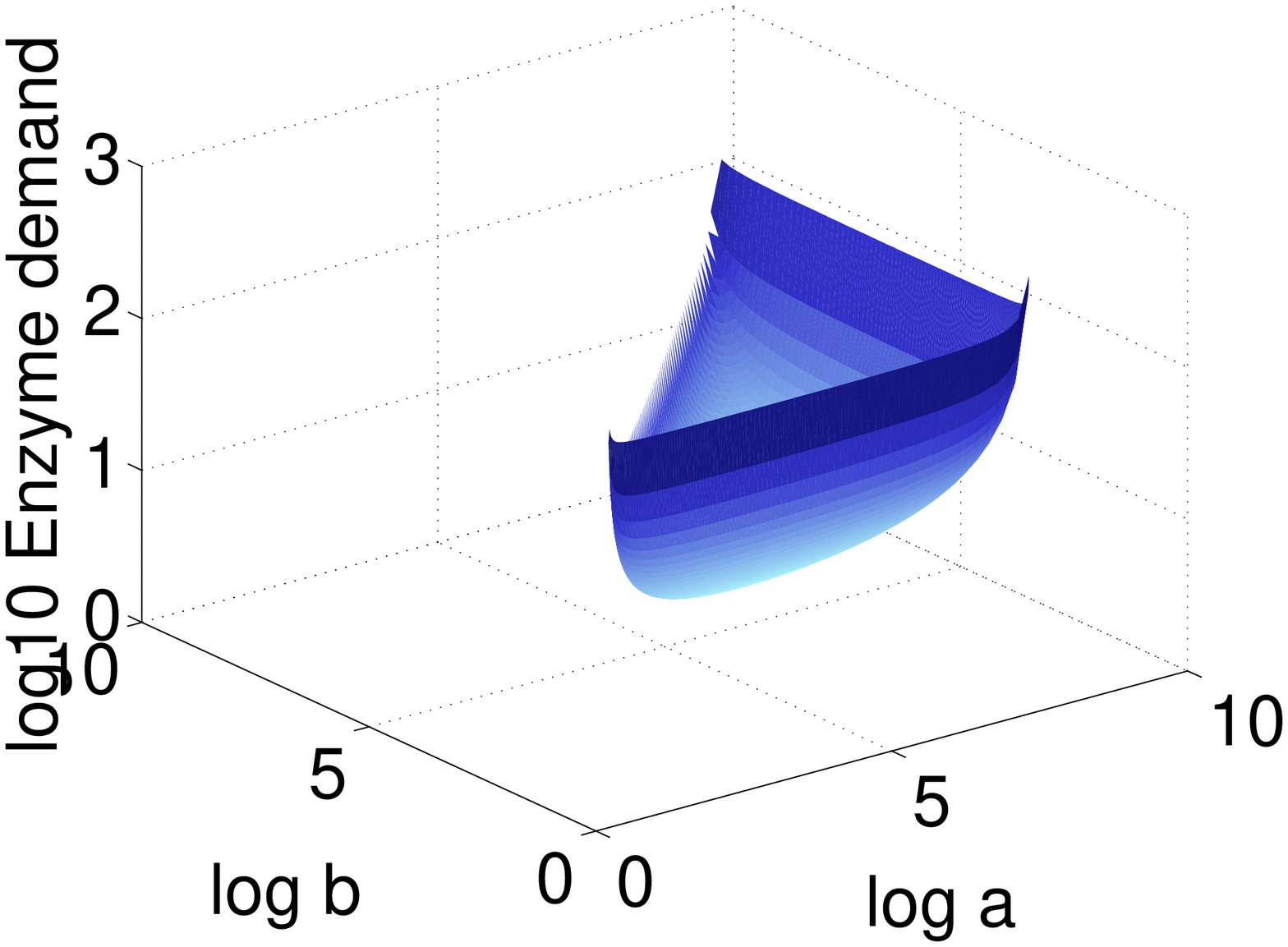}
 \end{tabular}
 \caption{Enzyme demand in a metabolic pathway.  (a) Pathway with
   reversible Michaelis-Menten kinetics (equilibrium constants and
   $\kM$ values are set to 1).  The external metabolite levels $x$ and
   $y$ are fixed, while internal levels $a$ and $b$ can vary.  Plots
   (b)-(d) show the enzyme demand for reactions 1, 2, and 3 (enzyme
   levels needed to sustain the desired unit flux). Regions of
   infeasible metabolite profiles are shown in grey.  At the edges of
   the admissible region (where chemical equilibrium would obtain),
   the thermodynamic driving force goes to zero; this must be
   compensated by a high enzyme level.  The enzyme demand in reaction
   1 (shown in (b)), for instance, increases with the level of A
   (x-axis) and goes to infinity as the mass-action ratio $a/x$
   reaches the equilibrium constant (where the driving force
   vanishes). (e) Total enzyme demand (sum of enzyme levels).  The
   metabolite polytope -- the intersection of feasible regions for all
   reactions -- is a triangle, and the enzyme demand is a cup-shaped
   function on this triangle.  The minimum point marks the optimal
   metabolite levels, from which optimal enzyme levels can be computed.}
  \label{fig:fourchain}
  \end{center}
\end{figure}

\section{The metabolite polytope represents the states of a kinetic model}
\label{sec:SIproofUniqueMapping}

\myparagraph{The metabolite polytope} A metabolic network (with given
flux directions, equilibrium constants, and metabolite bounds) defines
a convex \emph{metabolite polytope} $\metabolitepolytope$ in the space
of log-concentrations $x_i = \ln c_i$ (where $c_i$ is measured in
units of the standard concentration $c_\sigma=1$ mM).  An example is
shown in Figure \ref{fig:fourchain}.  In general, the polytope arises from two
sorts of inequality constraints: (i) Upper and lower bounds $x^{\rm
  min}_{i} \le x_{i} \le x^{\rm max}_{i}$ for metabolite levels yield
a box-shaped metabolite polytope; some metabolite levels may also be
constrained to a fixed value.  (ii) Reaction fluxes dissipate Gibbs
energy ($\Theta_l \cdot v_l>0$), so the driving forces must be
positive in the direction of the flux.  The resulting constraints $0 <
\Theta_l = \frac{1}{R T} \Delta_{\rm r}{G'}^{\circ}_{l}+ \sum_{i} \ln
c_{i}$ further restrict the metabolite polytope; they define E-faces
of the polytope (representing an equilibrium condition), where enzyme
costs rise steeply.  The metabolite polytope is a convex polytope in
log-concentration space, which contains all feasible metabolite
profiles.  It is bounded by two types of faces: faces that represent
an equilibrium in one of the reactions (``E-face''), where enzyme cost
goes to infinity; and faces stemming from physiological metabolite
bounds (``P-face''). Minimum points of the enzyme cost function can be
inside the polytope or on P-faces.  The polytope's precise shape
depends on the Enzyme Cost Function (ECF) score chosen (i.e, on the
simplifications applied) and on rate laws, rate constants, and
specific enzyme costs $\hEl$ in the model.

The metabolite polytope plays a central role in enzyme cost
minimization: For a given model and flux profile $\vv$, the points of
the polytope parametrize the set of all possible steady states
$(\ev,\cv,\vv)$.  Feasible metabolite profiles (represented by
polytope points) can be uniquely mapped to enzyme profiles, while the
mapping from enzyme to metabolite profiles need not be unique.  
The entirety of metabolic states (or all steady states) of a
kinetic model can be parametrized as follows: we consider the
(non-convex) flux polytope and construct, for each point, the
metabolite polytope. The construction yields all steady states
(characterized by concentrations, fluxes, enzyme levels) \emph{exactly
  once} (while the same enzyme profile may appear several times).
Using this fact, we can parametrize all metabolites states of a
kinetic model in a simple and systematic way (for details, see appendix \ref{sec:SIparametrization}). However, a restriction
to \emph{stable} steady states is not easily possible.

\section{Enzyme cost is a convex function on the metabolite polytope}

\myparagraph{Enzyme cost functions are convex on the metabolite
  polytope} The enzyme cost functions (\ref{eq:ECF}) and 
 (\ref{eq:TotalEnzymeDemand}) are differentiable convex functions
on the metabolite polytope (proof in SI \ref{sec:convexityProof}).
Convexity means that an interpolated metabolite vector, on a line
between two log-concentration vectors $\xv_{\rm a}$ and $\xv_{\rm b}$,
has a cost that is higher than (or at most equal to) the interpolated
cost:
\begin{eqnarray}
\label{eq:ExplainConvexity}
  \forall \lambda \in [0, 1]: \yy_{\rm pw}(\lambda \,\xv_{\rm a} +
  (1-\lambda)\,\xv_{\rm b}) \leq \lambda \, \yy_{\rm pw}(\xv_{\rm a}) +
  (1-\lambda)\,\yy_{\rm pw}(\xv_{\rm b}).
\end{eqnarray}
To show that the ECF scores are convex, we consider the most general
rate laws with denominator (\ref{eq:DECF4}) and rewrite it in the
form
\begin{eqnarray}
v &=&  \enzyme \cdot  \kcatplus \cdot  \eta^{\rm th} \cdot \eta^{\rm kin},
\end{eqnarray}
implying the  enzyme cost function
\begin{eqnarray}
\label{eq:myLittleCostFunction}
\yy = \sum_{l} y_{l} = \sum_{l} \frac{\hEl\,v_{l}}{\kcatl} 
\cdot \frac{1}{\eta^{\rm th}_{l}} \cdot\frac{1}{\eta^{\rm kin}_{l}}
\end{eqnarray}
for a pathway. The efficiency terms are given by
\begin{eqnarray}
\label{eq:SIConvexEtaDefinitions}
\eta^{\rm th} &=& 1 - \e^{-\Theta} = 1 - 
\exp \left(\frac{1}{RT}\,\Delta_{\rm r} {G^\circ}' + \sum_i n_i \ln c_i\right) \nonumber \\
\eta^{\rm kin} &=& 
\prod_i \left(\frac{s_i}{\kMi}\right)^{-m^{\rm S}_i}
\left(\sum_{k} M_{k} \prod_j c_{i}^{m_{ik}}\right)\inv
=\left(\sum_{k} \alpha_{k} \prod_j c_{i}^{a_{ik}}\right)\inv
\end{eqnarray}
 with coefficients $\alpha_k \in \Rset_+$ and $a_{ik}\in \Rset$.  The
 regulation efficiency $\eta^{\rm reg}$ can be neglected because it
 can always be covered by the term $\eta^{\rm kin}$.  
The cost function (\ref{eq:myLittleCostFunction}) with efficiency
terms (\ref{eq:SIConvexEtaDefinitions}) is convex on the metabolite
polytope.  The function stays convex if the investment function
$H(\ev)$ is not linear, but convex.  Importantly, even though all ECF
scores are convex, they may not be strictly convex (in which case
there would be a $<$ sign, instead of $\leq$, in
Eq.~(\ref{eq:ExplainConvexity})). For instance, simplified ECF scores
can be constant in the metabolite polytope.  Non-strict convexity can
arise when the mapping from enzyme to metabolite profiles is not
unique.  However, it is possible to enforce a unique optimum by adding
a convex regularization term $\regterm$, e.g., a quadratic function
favoring metabolite levels in the center of the typical concentration
range.  Such terms can be justified by biological side objectives: for
instance, keeping metabolite levels away from their upper or lower
bounds in advance will later allow cells to vary them more flexibly.
We can also consider a variant of ECM with an additional
metabolite-dependent objective $z$.  Instead of minimizing the enzyme
cost alone, we then minimize the difference $\yy_{\rm pw}^{\rm
  eff}(\xv) = \yy_{\rm pw}(\xv)-z^{\rm met}(\xv)$.  For instance, an
objective $z^{\rm met}(\ln \cv) = \sum_i \gamma_i\,(\ln c_i - \ln
\hat{c}_i)^2$ would represent a preference for metabolite levels close
to ``ideal'' levels $\hat{c}_i$, with cost weights $\gamma_i$.  If
$z^{\rm met}(\xv)$ is strictly concave, $-z$ can be used as a
regularization term $\regterm$. The resulting effective cost $\yy_{\rm
  pw}^{\rm eff}$ will be strictly convex even if the term $z$ is very
small.

\myparagraph{Enzyme cost minimization as an optimality problem for
  metabolite levels} The convexity proof suggests that enzyme levels
can be predicted with relatively little effort. Enzyme cost
minimization uses a metabolic network, a given flux profile $\vv$, and
possibly kinetic rate laws (with their thermodynamic or kinetic
constants), and specific enzyme costs.  The flux profile may be
stationary (like flux profiles determined by FBA) or non-stationary
(like experimentally measured fluxes, simply inserted into a
model). In any case, it must be free of thermodynamically unfeasible
cycles, and must agree with the assumed equilibrium constants and
external metabolite levels. If
the given flux directions are infeasible, the metabolite polytope
will be an empty set.  To find an optimal state, we choose an ECF
score and minimize the total enzyme cost within the polytope.  Optimal
metabolite profiles, enzyme profiles, and enzyme costs are obtained by
solving the Enzyme Cost Minimization (ECM) problem
\begin{eqnarray}
\label{eq:ECMproblem}
\fluxcost(\vv)   &=& \mbox{min}_{\xv \in \metabolitepolytope}\, \yy_{\rm pw}(\xv)\nonumber \\
\xv^{\rm opt}(\vv) &=& \mbox{argmin}_{\xv \in \metabolitepolytope}\, \yy_{\rm pw}(\xv)\nonumber \\
\ev^{\rm opt}(\vv) &=& \ev(\vv,\cv^{\rm opt}(\vv))
\end{eqnarray}
for log-concentration vectors $\xv = \ln \cv$.  The total cost
$\yy_{\rm pw}(\xv)$ (see Eq.~(\ref{eq:TotalEnzymeDemand})) is a sum
of enzyme costs (\ref{eq:ECF}) or simplified ECF scores. If there is
no unique optimum for $\xv$ (because the cost function is constant
along some subspace, and therefore not structly convex), a unique solution can be enforced by adding a
convex regularization term $\regterm(\xv)$ to $\yy_{\rm pw}$.
Since the optimal enzyme levels depend
on external metabolite levels, they must be recalculated after changes
in external conditions.  If non-enzymatic reactions (typically with
mass-action rate laws) are included in the optimality problem, they
contribute to the energetic constraints, but not to the enzyme cost
function.

\section{Discussion}

\myparagraph{Metabolite levels are suitable variables for cost
  optimization} In summary, we saw that logarithmic metabolite levels
are suitable variables for screening, sampling, and optimizion of
metabolic states. Due to the mapping from metabolite profiles to
enzyme profiles, all feasible enzyme profiles can be reached, where
bounds on driving forces can be formulated as linear
constraints. Moreover, using the metabolite log-concentrations as free
variables does not only provide a good search space, but also
facilitates the optimization problem: under general and reasonable
assumptions, the ECF scores are convex functions on the metabolite
polytope.  Convexity holds for a wide range of rate laws, including
rate laws with allosteric regulation. As a consequence, the
optimization remains tractable for various rate laws and larger
metabolic networks.  Strict convexity (required for an isolated
optimum point) can be enforced by adding small regularization terms
$\regterm$, possibly representing biological side objectives.

\begin{figure}[t!]
  \begin{center}
    \includegraphics[width=9cm]{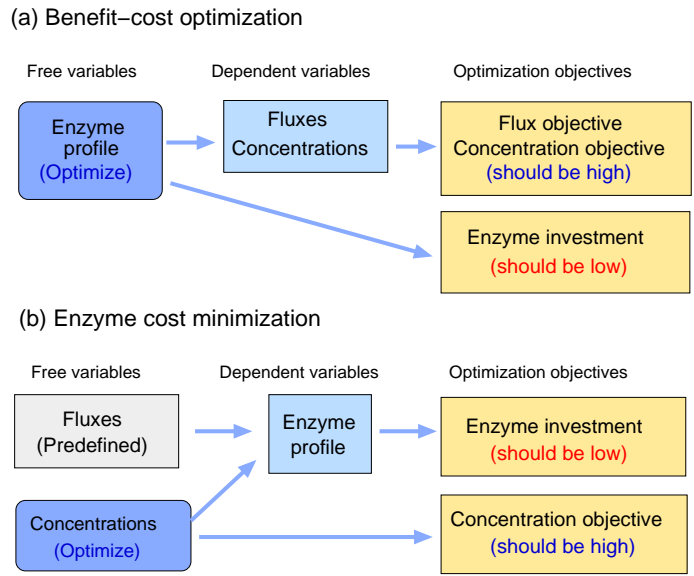}
  \end{center}
    \caption{Two ways of framing enzyme allocation as an optimality
      problem. (a) Benefit-cost optimization.  Each enzyme profile
      determines a metabolic state (with state variable vectors $\vv$
      and $\cv$ and a metabolic objective $z(\vv,\cv)$) and an
      investment $h(\enzyme_1, \enzyme_2, ..)$. To predict an optimal
      enzyme profile, we maximize the difference $z-h$. (b) In ECM, a
      predefined flux profile is realized by an enzyme profile (and a
      corresponding metabolite profile) with a minimal investment; a
      concentration objective can be considered in addition.  Using
      the metabolite concentrations as free variables makes the
      problem relatively easy to solve. }
  \label{fig:optimalityProb}
\end{figure}

\myparagraph{Relationship to other optimality-based metabolic models}
Optimal enzyme allocation in kinetic models can be framed in two main
ways. One the one hand, enzyme levels can be treated as control
variables which determine the metabolic state, and thus the fluxes
(see, e.g., \cite{reic:83,hekl:1996}). The aim is to find the enzyme
profile that leads to an optimal state (where enzyme cost can come
into play as a constraint or as a penalty functions).  On the other
hand (as in \cite{fnbl:13,tnah:13} and here), one can predefine the
fluxes and search for enzyme levels that realize them in an optimal
way (here, minimal cost is used as the optimality criterion).  Both
approaches address similar problems and lead to equivalent
solutions. For instance, if we first maximize a flux at a fixed total
enzyme level, as in \cite{hekl:1996}, and then use this flux as a
constraint in ECM (with identical specific costs for all enzymes), we
recover the metabolite and enzyme profile from the initial
calculation. In fact, both optimality problems can be derived from a
common general optimality problem by constraining the fluxes or the
total enzyme cost.  The approaches frame the same problem, but in
different ways.  Using fluxes as a scaffold for model construction has
several advantages. First, it makes the optimality problem convex.
Second, we can precisely specify the flux state to be modelled. Third,
the flux cost functions $\fluxcost(\vv)$ can be used in flux analysis.
Thus, ECM enables more realistic variants of FBA; the cost function
may contain additional side objectives scoring the metabolite levels.
If flux profiles are compared at a given flux benefit, it is only the
cost scores that count in the optimization, so cost and benefit can
be measured in different units.  In benefit-cost approaches as in
Figure \ref{fig:optimalityProb} (a), enzyme investment and metabolic
objective are directly compared and show the same physical units.  In
order to make them comparable, a relative weighting would have to be
established, which bears the problem of arbitrariness.

\myparagraph{What other objectives matter besides cost-optimality, and
  how can they be included in ECM?}  Our fundamental assumption --
that enzyme levels are cost-optimized in every moment -- is of course
debatable.  Instead, proteins may be expressed to higher amounts to
anticipate sudden challenges (example: energy production in muscle
cells).  Preemptive protein expression can avoid costs for rearranging
the proteome and performance losses during adaptation; however, in ECM
it would appear futile. Also flexibility in metabolite concentrations
can be important, and cells might trade it against enzyme
economy. Furthermore, enzyme and metabolite levels in cells are not
only shaped by demands in a single pathway, but also by other pathways
outside the model in question. Finally, if proteins are used as an
amino acid storage, there will be little pressure to keep them at low
concentrations. How can our method be useful despite all this? First,
an account of simple enzyme economy can be a basis for studying more
complicated optimality requirements afterwards.  Second, despite all
these points of critique, enzyme economy may be the main requirement,
e.g., during fast, nitrogen-limited growth.  Third, we can study how
deviations from the optimal state affect enzyme cost, and thus
fitness.  Finally, ECM can be extended to include more objectives and
constraints into our pathway model, and thus to account for the
surrounding cell.  The metabolites from our pathway may also be
involved in other pathways outside the model. If these other pathways
demand higher or lower metabolite levels, we can implement this fact
in ECM by constraints (upper and lower concentration bounds) or by
concentration-dependent side objectives $z^{\rm met}(\cv)$, which
penalize unfavorable metabolite levels. Trade-offs between the
 pathway in focus and other pathways around it can be handled in this
way. Alternatively, we can assume that each metabolite level should
be close to the centre of its allowed range (which also provides
flexibility, because it will not easily hit a bound).  Again, this can
be captured by side objectives. 

\myparagraph{Kinetics-based flux cost functions for usage in FBA.}
As a possible application, enzyme
 cost functions allow us to define non-linear flux costs for use
 in flux balance analysis.  ECM is based on 
 a given flux profile. However, since it can be applied 
to \emph{any} flux profile, it defines a flux cost
function $\fluxcost(\vv)$, which can be applied in flux prediction.
FBA with minimal fluxes (fmFBA) compares flux profiles at equal
benefit (FBA objective) and minimizes their heuristic cost.  The flux
cost functions used are linear (for a predefined choice of  flux directions). Linearity
simplifies calculations, but is not very realistic: first, cost scores
like the sum of fluxes do not account for kinetics and regulation;
second, the costs add linearly when flux distributions are linearly
combined. Flux cost functions obtained from ECM, and based on a
kinetic model, are more realistic. In an mfFBA based on such cost
functions, one would predefine flux directions, flux bounds, and a
flux benefit $\hat{\zv}^{\rm v} \cdot \vv = \hat b$, and assume
stationary fluxes; but instead of a linear flux cost, one would
minimize the flux cost $\fluxcost(\vv)$.  Flux costs derived from ECM
are concave functions on the flux polytope (with given flux
directions). This implies that the solutions of the new fmFBA problems
will be elementary flux modes, which confirms findings from other
enzyme optimality approaches \cite{murs:14,wpht:14}. In fact, the flux
cost function $\fluxcost(\vv)$ can be expected to be strictly concave
(except for specific cases, e.g.~models containing two identical
reactions with identical rate laws). If this is this case, elementary
flux modes are the only solutions.  As a consequence, splitting
a  flux profile into elementary modes that run in different
compartments or at different time points can be better, but never 
worse than the original flux profile in terms of enzyme cost.

\section*{Acknowledgements}
The authors thank Avi Flamholz, Ron Milo, Frank Bruggeman, Joost
Hulshof, and Meike Wortel for inspiring discussions.
This work was supported by the German Research Foundation (Ll
1676/2-1).

\bibliographystyle{unsrt}
\bibliography{/home/wolfram/latex/bibtex/biology}

\begin{appendix}

\begin{table}[t!]
\begin{center}
  \begin{tabular}{|lll|}
\hline
    \cellcolor{lightblue} \textbf{Name} &
    \cellcolor{lightblue}  \textbf{Symbol}&
    \cellcolor{lightblue} \textbf{Unit}\\ \hline & & \\[-2mm]
    Flux & $v_{l}$& mM/s \\
    Metabolite level & $c_{i}$ & mM     \\
    Logarithmic metabolite level & $x_i = \ln (c_{i}/c_{\sigma})$ & mM     \\
    Enzyme level & $\el$ & mM\\
    Reaction rate & $v_{l}(\el, \cv) = \el \cdot \ratelaw_{l}(\cv)$ & mM/s \\
    Specific rate & $\ratelaw_{l} = v_{l}/\el$& 1/s\\
    Scaled reactant elasticity & $\Esc_{li}$ & 1 \\
    Gibbs free energy of formation & ${G'}^\circ_{i}$ & kJ/mol\\
    Reaction Gibbs energy & $\Delta_{\rm r} G'_{l} = \Delta_{\rm r} {G'}^\circ_{l} + \sum_{i} n_{il}  RT \,\ln c_{i}$ & kJ/mol\\
    Driving force & $\Theta_l = - \Delta_{\rm r} G'_{l}/RT$ & 1 \\
    Forward/backward catalytic constant  & $\kcatplus, \kcatminus$      & 1/s    \\
    Michaelis-Menten constant     & $\kMli$ & mM     \\
    Specific enzyme cost & $\hEl$ & D/mM \\
    Enzyme cost  & $y_{l} = \hEl\,\el$ & D \\
    Total enzyme cost & $\yy = \sum_{l } \hEl \,\el$ & D \\
    Specific flux cost & $\rrl$ & D/(mM/s) \\
    Enzyme-optimal cost & $\fluxcost(\vv) = \min_{\xv in \metabolitepolytope} \yy_{\rm pw}(\xv)$ & D\\
\hline
  \end{tabular}
\end{center}
\caption{Terms and symbols used in enzyme cost minimization.  Darwin
  (D) is a hypothetical fitness unit. Reaction directions are defined
  in such a way that fluxes are positive.  To define
  log-concentrations, we use the standard concentration $c_\sigma=$
  1mM (shown here, but omitted elsewhere for simplicity.) }
\label{tab:symbolsshort}
  \end{table}

\section{Kinetic rate laws}
\label{sec:SIratelaws}

\myparagraph{Simplified denominators} By considering simple enzyme
mechanisms with few binding states, we obtain general rate laws
applicable to all reaction stoichiometries. The rate law denominators
to be used in  Eq.~(\ref{eq:GeneralRateLawRate})
have simple structures (containing only few of the possible sum terms,
and with prefactors following from a few Michaelis-Menten constants)
\cite{liuk:10}.  If denominator terms are omitted, the rate will be
overestimated, i.e., enzyme demand and costs will be underestimated.
First, there are rate laws with denominators
\begin{eqnarray}
  \label{eq:ECF2ratelaws}
 D^{\rm (S)} &=& \prod_{i} (s_{i}/\kMi)^{m^{\rm S}_{i}} \nonumber \\
 D^{\rm (SP)} &=& 
 \prod_{i} (s_{i}/\kMi)^{m^{\rm S}_{i}} + \prod_{j} (p_{j}/\kMj)^{m^{\rm P}_{i}}
\end{eqnarray}
which lead to the energetics-based ECF2 scores.  The big product terms
are called principal substrate and product terms. As before, $s_i$ and
$p_j$ denote substrate and product levels.  The first formula assumes
that substrate levels are high and product levels are low; the second
one assumes that substrate and product levels are both high. Next,
there are rate laws with denominators
\begin{eqnarray}
  \label{eq:ECF3ratelaws}
 D^{\rm (1S)} &=& 
1 +  \prod_{i} (s_{i}/\kMi)^{m^{\rm S}_{i}}\nonumber \\
 D^{\rm (1SP)} &=& 
1 + \prod_{i} (s_{i}/\kMi)^{m^{\rm S}_{i}}+ \prod_{j} (p_{j}/\kMj)^{m^{\rm P}_{i}}
\end{eqnarray}
which lead to the saturation-based ECF3 scores.  The denominators
contain only three possible terms: the term 1, the principal substrate
term, and the principal product term. To justify these rate laws, we
assume a strongly cooperative binding between substrates and between
products.  The first formula assumes low product concentrations; the
second formula describes the direct-binding modular rate law
\cite{liuk:10}.  The direct-binding modular rate law is a generalized
version of reversible MM kinetics.  In the underlying enzyme
mechanism, the enzyme exists in three states: fully bound with
substrates, fully bound with products, or empty.  
\myparagraph{Allosteric regulation}  
If enzymes are allosterically regulated, the rate law denominators
contain additive or multiplicative terms for regulation
\cite{liuk:10}.  Additive terms can arise from competitive
regulation. Multiplicative terms (for non-competitive regulation) can
be split from the denominator and treated as prefactors in the rate
law. Typical choices are $\frac{x}{x+k^{\rm A}_{X}}$ for
non-competitive activation and $\frac{k^{\rm I}_{X}}{x+k^{\rm I}_{X}}$
for non-competitive inhibition, with rate constants $k^{\rm A}$ and
$k^{\rm I}$ and regulator concentration $x$ \cite{liuk:10}.
Accordingly, allosteric effects can either be listed by a separate
efficiency term in the factorized ECF formulae, or be included in the
kinetic efficiency. For instance, the kinetic efficiency term for
MM-kinetics with non-competitive inhibition can be split into
\begin{eqnarray} 
 \eta^{\rm kin} = \frac{s/\kMS}{(1+\frac{x}{K_{\rm I}}) (1+\frac{s}{\kMS}+\frac{p}{\kMP})}  = 
\frac{1}{1+\frac{s}{\kMS}+\frac{p}{\kMP}} \frac{1}{(1+x/K_{\rm I})} = 
\eta^{\rm kin*}\, \eta^{\rm reg}.
\end{eqnarray} 

\section{Factorization of rate laws}
\label{sec:SIfactorization}

 \myparagraph{Separable rate laws and efficiency term} To demonstrate
 how rate law are factorized, we consider the common modular rate (CM)
 law \cite{likl:06a, liuk:10}, a generalized form of reversible MM
 kinetics with the denominator
\begin{eqnarray}
  \label{eq:ECF3ratelaws}
 D^{\rm (CM)} &=& 
 \prod_{i} (1+s_{i}/\kMi)^{m^{\rm S}_{i}}+ \prod_{j} (1+p_{j}/\kMj)^{m^{\rm P}_{i}} -1.
\end{eqnarray}
In the assumed enzyme mechanism, substrate molecules bind
independently, product molecules bind independently, and substrate and
product binding exclude each other.  For a bimolecular reaction $A + B
\rightleftharpoons P + Q$, the rate law
 \begin{eqnarray}
 \label{eq:convenienceKinetics}
 v &=& \enzyme\, \frac{\kcatplus\, \frac{[A][B]}{K_A\,K_B} - \kcatminus\, \frac{[P][Q]}{K_P\,K_Q}}
 {(1+\frac{[A]}{K_A})(1+\frac{[B]}{K_B})  + (1+ \frac{[P]}{K_P}((1 + \frac{[Q]}{K_Q}) -1} 
 \end{eqnarray}
can be rewritten as 
 \begin{eqnarray}
 \label{eq:convenienceKinetics2}
 &=& \enzyme\, \kcatplus\, \frac{\frac{[A][B]}{K_A\,K_B} - \frac{\kcatminus}{\kcatplus}\frac{[P][Q]}{K_P\,K_Q}}
 { (1 + \frac{[A]}{K_A} + \frac{[B]}{K_B} + \frac{[A][B]}{K_{AB}} + \frac{[P]}{K_P} + \frac{[Q]}{K_Q} + \frac{[P][Q]}{K_{PQ}})}\nonumber \\
 &=& \enzyme\, \kcatplus\, \frac{1 - \e^{-\Theta}} 
 {\frac{K_A\,K_B}{[A][B]} (1 + \frac{[A]}{K_A} + \frac{[B]}{K_B} + \frac{[A][B]}{K_{AB}} + \frac{[P]}{K_P} + \frac{[Q]}{K_Q} + \frac{[P][Q]}{K_{PQ}})}
 \nonumber \\
 &=& \enzyme\, \kcatplus\, [1 - \e^{-\Theta}] \eta^{\rm kin}
 \end{eqnarray}
 where we defined the kinetic efficiency
 \begin{eqnarray}
  \eta^{\rm kin} &=& \frac{1}
 {\frac{K_A\,K_B}{[A][B]} (1 + \frac{[A]}{K_A} + \frac{[B]}{K_B} + \frac{[A][B]}{K_{AB}} + \frac{[P]}{K_P} + \frac{[Q]}{K_Q} + \frac{[P][Q]}{K_{PQ}})}
 \end{eqnarray}
 and used the Haldane relationship $K_{eq} =
 \frac{\kcatplus}{\kcatminus} \frac{K_P\,K_Q}{K_A\,K_B}$ and the
 identity $\e^{-\Theta} = \frac{[P][Q]}{[A][B]}/K_{\rm eq}$.  In the
 calculation, we first separated the $k^+_{\rm cat}$ value from the
 rest of the fraction, and then hid the negative flux term in the
 energetic efficiency term $\eta^{\rm th}$.

\section{Parametrizing all states of a kinetic model}
\label{sec:SIparametrization}

In a kinetic model with given rate laws and external metabolite
concentrations, an enzyme profile $(\enzyme_1,.. \enzyme_2, ..)$ lead
to a steady state with metabolite levels $\cv$ and fluxes $\vv$. The
following proposition shows how the set ${\mathcal S}$ of such steady
states $\sigma = (\ev, \cv\, \vv)$ can be easily parametrized.

\textbf{Proposition:} Consider a kinetic model with rate laws $v_{l}=
\el\,\ratelaw_{l}(\cv)$, thermodynamically consistent rate constants
(satisfying Wegscheider conditions and Haldane relationships), a
feasible positive flux profile $\vv$, and bounds on metabolite
levels. Any feasible metabolite profile $\ln \cv \in
\metabolitepolytope$ can be realized by some positive enzyme profile
$\ev$; given the metabolite levels, the enzyme levels are uniquely
determined and given by $\el(\ln \cv) = v_{l}/\ratelaw_{l}(\cv)$,
which is a differentiable function on the metabolite polytope.

\textbf{Proof:} If a metabolite profile $\cv$ is feasible for the
given flux profile $\vv$, the specific rates $\ratelaw_{l}(\cv)$
obtained from reveresible rate laws (see
Eq.~(\ref{eq:GeneralRateLawRate}) in appendix) have the same signs as
$v_{l}$, so $\el = v_{l}/\ratelaw_{l}(\cv)$ is positive on the entire
metabolite polytope. Since $\ratelaw_{l}(\cv)$ is differentiable and
does not change its sign on the metabolite polytope, $\el(\ln \cv)$ is
differentiable on the metabolite polytope. 

According to our proposition, any thermodynamically feasible
metabolite profile can be realized by some steady state of the kinetic
model (with an appropriate choice of enzyme levels), so the set
${\mathcal S}$ of metabolic states with a given flux profile $\vv$ can
be characterized by points of the metabolite polytope.  In particular,
the set of kinetically realizable metabolite profiles depends on the
equilibrium constants, but not on enzyme-specific rate constants.

With simplified rate laws, the same enzyme profile may be realizable
by different metabolite profiles. (i) If a metabolite appears in a
model but has no impact on any reaction, its concentration can be
freely varied, independently of the enzyme levels.  (ii) With
simplified cost scores in which all efficiencies $\eta^{\rm th}$,
$\eta^{\rm kin}$, and $\eta^{\rm reg}$ are taken to be constant,
enzyme levels do not depend on metabolite levels. (iii) With
simplified scores in which $\eta^{\rm kin}$ and $\eta^{\rm reg}$ are
taken to be constant, enzyme costs depend on metabolite levels
\emph{only through the thermodynamic forces}. Notably, the vector $\xv
= \ln \cv$ can be varied along directions in the nullspace of
$\Ntot\trans$ without affecting the driving forces or enzyme
cost. Thus, the enzyme cost scores have an invariant subspace on the
metabolite polytope (namely the nullspace of $\Ntot\trans$).  Under
what conditions more complicated enzyme cost scores (without
regularisation terms) have unique optima remains an open question.

Finally, to parametrize \emph{all} steady states of a kinetic model,
we can follow a two-step procedure in which we enumerate all possible
flux distributions and, for each of them, all possible enzyme and
metabolite profiles. The thermodynamically feasible flux distributions
$\vv$ form a set ${\mathcal V}$, given by ${\mathcal V} = \{\vv |
\exists \xv: \sign(\vv) =\sign(-\Delta_{\rm r} G'(\xv) \}$, where
$\xv$ stands for log-concentration profiles. The reaction Gibbs
energies $\Delta_{\rm r} G'_l = \Delta_{\rm r} {G'}^\circ_l + RT
\sum_i n_{il}\,\ln c_i$ depend on the internal and external metabolite
levels and on the equilibrium constants chosen. According to
thermodynamic condition, whether a flux distributions is feasible or
infeasible depends solely on its sign pattern.  By imposing upper and
lower bounds and the stationarity condition, we can further limit this
set and obtain the set of feasible, stationary fluxes ${\mathcal
  V}^{\rm stat} = \{\vv | \vv^{\min} \le \vv \le \vv^{\rm max} \wedge
\Nint\,\vv = 0 \wedge \exists \xv: \sign(\vv) =\sign(-\Delta_{\rm r}
G(\xv) \}$. ${\mathcal V}^{\rm stat}$ is a (generally non-convex)
polytope in flux spaces.  Each flux distribution $\vv$ in this set
defines a set of possible states ${\mathcal S}_{\vv}$, one can then
set ${\mathcal S} = \{(\ev,\cv,\vv) | \vv \in {\mathcal V} \wedge
(\ev,\cv) \in {\mathcal S}_{\vv} \}$.
 
\section{Convexity proof for enzyme cost functions}
\label{sec:convexityProof}

To prove the convexity of general enzyme cost functions, on the
metabolite polytope and at given desired fluxes, we first show the
convexity of some simple functions appearing in the formula.

\subsection{General lemmata}
\begin{lemma}\label{lemma:exp_cvx} The function 
$f(y) = -\ln(1-\e^y)$ is convex in the range $y < 0$.
\end{lemma}
\begin{proof} The second derivative 
\[\frac{\mbox{d}}{\mbox{d}y^2} f(y) = \frac{\e^y}{(1-\e^y)^2}\]  is positive for $y < 0$.
\end{proof}

\begin{lemma}\label{lemma:lse_cvx} The function 
$f(\xv) = \ln \sum_{k=1}^{n} \e^{x_k}$ is convex.
\end{lemma}
\begin{proof}

\[\nabla^2 f(\xv) = \frac{\diag(\cv) ({\bf 1}\trans \cv) - \cv\,\cv\trans}{({\bf 1}\trans \cv)^2} ~~~~~~~~ (\text{where } c_i = e^{x_i})\]

\[\forall \uv:~~ \uv \trans \nabla^2 f(\xv) \uv = 
  \frac{(\sum_i c_i u_i^2)(\sum_i c_i)- (\sum_i u_i c_i)^2}{(\sum_i c_i)^2} \ge 0
\]
since $(\sum_i u_i\, c_i)^2 \le (\sum_i c_i\, u_i^2)(\sum_i c_i)$ from
the Cauchy-Schwarz inequality. Therefore, the Hessian $\nabla^2 f(\xv)$ is
positive semi-definite, which proves that  $f(\xv)$ is convex.

\end{proof}

\begin{lemma}\label{lemma:th_cvx}
For any number  $\nu \in \mathbb{R}_+$ and vector $\nv \in \mathbb{R}^m$, the
function $- \ln (1 - \nu\, \e^{\nv\cdot \xv} )$ is convex over $\{\xv \in
\mathbb{R}^m~|~\nu\, \e^{\nv \cdot \xv} < 1\}.$
\end{lemma}
\begin{proof}
 This function is a composition of $f = -\ln(1-\e^y)$ from Lemma
 \ref{lemma:exp_cvx} with the affine transformation $y = \nv \cdot \xv
 + \ln{\nu}$, an operation which preserves convexity.
\end{proof}

\begin{lemma}\label{lemma:kin_cvx}
For any matrix ${\bf A} \in \mathbb{R}^{n \times m}$ and vectors ${\bf
  b} \in \mathbb{R}^{n}_+$, the following function is convex over
${\bf x} \in \mathbb{R}^{m}$:
\begin{eqnarray}
\ln \left( \sum_{k=1}^{n} \e^{\av_{k}\cdot\xv + b_k} \right)
\end{eqnarray}
where $\av_{i}$ is the $i$th row of $\bf A$.
\end{lemma}

\begin{proof}
 This function is a composition of $f = \ln \sum_{i=1}^{n} \e^{x_i}$
 from Lemma \ref{lemma:lse_cvx} with the affine transformation $x_i =
 \av_{i}\cdot\xv + b_i$, an operation which preserves convexity.
\end{proof}

\subsection{The convexity of enzyme cost functions}

\begin{lemma}
Assume that all enzyme-catalysed reactions in a model behave according
to rate laws of the type
\begin{eqnarray}
v &=&  \enzyme \cdot  \kcatplus \cdot  \eta^{\rm th} \cdot \eta^{\rm kin},
\end{eqnarray}
with $\eta^{\rm th}$ and $\eta^{\rm kin}$ given by 
Eq.~(\ref{eq:SIConvexEtaDefinitions}), with coefficients $\alpha_k \in
\Rset_+$ and $a_{ik}\in \Rset$. Assume that the enzyme cost function
for enzymatic reaction $l$ reads
\begin{eqnarray}
y_{l} = \frac{\hEl\,v_{l}}{\el} = \frac{\hEl\,v_{l}}{\kcatl} 
\cdot \frac{1}{\eta^{\rm th}_{l}}
\cdot\frac{1}{\eta^{\rm kin}_{l}}\,.
\end{eqnarray}

Then the total enzyme cost $\yy = \sum_{l} y_{l}$, as a function
of logarithmic metabolite concentrations ($\xv = \ln \cv$), is convex.
\end{lemma}

\begin{proof}
To simplify the efficiency terms, we can use the abbreviations $x_{i} \equiv
\ln c_{i}$, $\nu \equiv \exp (\Delta_{\rm r} {G^\circ}'/RT)$, and $b_k = \ln
\alpha_k$:
\begin{eqnarray}
\eta^{\rm th} &=& 1 - \nu\,\e^{-\nv\cdot \xv} \nonumber \\
\eta^{\rm kin} &=& \left(\sum_{k=1}^{n} \e^{\av_{k}\cdot\xv + b_k}\right)\inv.
\end{eqnarray}

If we look at the natural logarithm of $y_{l}$,
\begin{eqnarray}
\ln y_{l} = \ln \left( \frac{\hEl\,v_{l}}{\kcatl}  \right) -\ln \eta^{\rm th}_{l} - \ln \eta^{\rm kin}_{l},
\end{eqnarray}
we see that each of the three terms in the sum is convex in $\xv$. The
first term is constant with respect to the metabolite concentrations
and therefore trivially convex.  The energetic term, $-\ln \eta^{\rm
  th} = - \ln (1 - \nu\,\e^{-\nv\cdot \xv})$, is convex according to
Lemma \ref{lemma:th_cvx}. The kinetic term, $-\ln \eta^{\rm kin} = \ln
\left(\sum_{k=1}^{n} \e^{\av_{k}\cdot\xv + b_k}\right)$, is convex
according to Lemma \ref{lemma:kin_cvx}. We conclude that $y_{l}$ is
convex too, since it is a composition of a convex function ($\e^x$)
with another convex function ($\ln y_{l}$).  Finally, the total enzyme
cost ($\yy$) is convex since it is a sum of convex functions:
\begin{eqnarray}
\yy = \sum_{l} y_{l}(\xv).
\end{eqnarray}
\end{proof}

\end{appendix}

\end{document}